\documentclass[10pt,twocolumn,twoside]{ieeeconf}
\usepackage{graphicx}
\usepackage{url}
\usepackage{xcolor}
\usepackage{amsmath}
\usepackage{mathtools}
\usepackage{tikz}
\usepackage{verbatim}
\usepackage{subcaption}
\usepackage{pgfplots}
\pgfplotsset{compat=1.10}
\usepackage{graphicx}
\usepackage{epstopdf}
\usepackage{amssymb}
\usepackage{relsize}
\usepackage[textwidth=3.8em,textsize=scriptsize,disable]{todonotes}
\usepackage{algorithm}
\definecolor{Gray}{gray}{0.90}
\usepackage{colortbl}

\usepackage{algorithm}
\usepackage{algorithmicx}
\usepackage{algcompatible}
\usepackage{algpseudocode}
\usepackage{color}
\usepackage{blkarray}

\usepackage{booktabs}

\newtheorem{theorem}{\bf{Theorem}}[section]

\newtheorem{conjecture}[theorem]{\bf{Conjecture}}

\newtheorem{prop}[theorem]{\bf{Proposition}}

\newtheorem{fact}[theorem]{\bf{Fact}}
\newtheorem{remark}[theorem]{\bf{Remark}}
\newenvironment{definition}[1][Definition]{\begin{trivlist}
\item[\hskip \labelsep {\bfseries #1}]}{\end{trivlist}}

\usepackage{amssymb}
\usepackage{url}
\newcommand{\calG}[0]{\mathcal{G}}
\newcommand{\calV}[0]{\mathcal{V}}
\newcommand{\calE}[0]{\mathcal{E}}

\newcommand{\calA}[0]{\mathcal{A}}
\newcommand{\calB}[0]{\mathcal{B}}

\newcommand{\calN}[0]{\mathcal{N}}

\newcommand{\calC}[0]{\mathcal{C}}

\newcommand{\norm}[1]{\left\lVert#1\right\rVert}
\definecolor{amber}{rgb}{1.0, 0.75, 0.0}
\begin{document}

\title{Interplay Between Resilience and Accuracy in Resilient Vector Consensus\\ in Multi-Agent Networks}
\author{Waseem Abbas, Mudassir Shabbir, Jiani Li, and Xenofon Koutsoukos
\thanks{W. Abbas, J. Li, and X. Koutsoukos are with the Electrical Engineering and Computer Science Department at Vanderbilt University, Nashville, TN, USA (Emails: \{\texttt{waseem.abbas, jiani.li, xenofon.koutsoukos\}@vanderbilt.edu)}. M. Shabbir is with the Computer Science Department at the Information Technology University, Lahore, Pakistan (Email: \texttt{mudassir@rutgers.edu)}.}
\thanks{This research was supported in part by the National Institute of Standards and Technology under Grant 70NANB18H198,
and by the National Science Foundation under award CNS1739328.}
}

\maketitle
%\IEEEtitleabstractindextext{%
%==================== Abstract =====================
\begin{abstract}
In this paper, we study the relationship between resilience and accuracy in the resilient distributed multi-dimensional consensus problem. We consider a network of agents, each of which has a state in $\mathbb{R}^d$. Some agents in the network are adversarial and can change their states arbitrarily. The normal (non-adversarial) agents interact locally and update their states to achieve consensus at some point in the convex hull $\calC$ of their initial states. This objective is achievable if the number of adversaries in the neighborhood of normal agents is less than a specific value, which is a function of the local connectivity and the state dimension $d$. However, to be resilient against adversaries, especially in the case of large $d$, the desired local connectivity is large. We discuss that resilience against adversarial agents can be improved if normal agents are allowed to converge in a bounded region $\calB\supseteq\calC$, which means normal agents converge at some point close to but not necessarily inside $\calC$ in the worst case. The accuracy of resilient consensus can be measured by the Hausdorff distance between $\calB$ and $\calC$. As a result, resilience can be improved at the cost of accuracy. We propose a resilient bounded consensus algorithm that exploits the trade-off between resilience and accuracy by projecting $d$-dimensional states into lower dimensions and then solving instances of resilient consensus in lower dimensions. We analyze the algorithm, present various resilience and accuracy bounds, and also numerically evaluate our results.
\end{abstract}

%========== Section: INTRODUCTION ===================
\section{Introduction}
\label{sec:Intro}
Consider a network of agents in which each agent maintains a $d$-dimensional state vector and updates it by interacting with a subset of other agents. Some of the network agents may be adversarial (or faulty) and therefore send incorrect states to their neighbors. Moreover, non-adversarial, or commonly referred to as the \emph{normal} agents, are unaware of these adversarial agents' identities. Resilient distributed multi-dimensional consensus problem requires that in the presence of adversarial agents, normal agents update their states to converge to a common state in the convex hull of normal agents' initial states. Resilient multi-dimensional or vector consensus has several applications, such as in multirobot networks where robots operate in a multi-dimensional workspace \cite{park2017fault,park2018robust}, distributed computing \cite{tseng2013iterative,vaidya2013byzantine,vaidya2014iterative}, distributed optimization \cite{Wang2019,sundaram2018distributed}, and fault-tolerant multiagent networks \cite{su2016fault,guerrero2018design}.

There are distributed algorithms achieving resilient consensus in networks under certain conditions, which include bounding the maximum number of adversaries in the neighborhood of each normal agent, for instance, \cite{leblanc2013resilient,dibaji2017resilient,abbas2017improving,usevitch2019resilient,saldana2017resilient,mendes2013multidimensional,dibaji2015consensus,ghawash2019leveraging}. A recently proposed approximate distributed robust convergence (ADRC) algorithm \cite{park2017fault} guarantees convergence if each normal agent $i$ has at most $\lceil N_i/2^d\rceil-1$ adversaries in its neighborhood, where $N_i$ is the size of the neighborhood of $i$, and $d$ is the dimension of the state vector. Thus, in a multirobot network in which the state of each robot is its position in a $3$-dimensional Euclidean space, each robot needs to have at least $9$ robots in its neighborhood to become resilient to a single adversary. We observe that the underlying network graph needs to be dense and highly connected to ensure resilient vector consensus. Resilience incurs significant overhead in the form of many interactions between agents, which explodes with an increase in the dimension $d$.% of the state of agents.
%Thus, resilience incurs a large overhead in terms of the number of local interactions that each agent needs to have within a network, and this number explodes with an increase in the dimension $d$ of the state of agents.

In this paper, we study the interplay between resilience and accuracy in the distributed vector consensus algorithms. We analyze algorithms' performance when the number of adversaries exceeds the allowed limit, which is a function of the state dimension $d$ and the size of the neighborhood of normal agents in the underlying network. We discuss that if the conditions in resilient vector consensus algorithms are not satisfied, then the adversary can drive normal agents arbitrarily far away from the convex hull of their initial positions. However, the local connectivity requirements of normal agents within the network can be significantly relaxed if we desire normal agents to converge at some point close to but not necessarily inside the convex hull of their initial positions. In other words, we can improve the resilience of vector consensus at the cost of accuracy. We adopt a simple approach of partitioning a $d$-dimensional state into multiple lower-dimensional states to explore this resilience-accuracy trade-off. Instead of solving a single instance of $d$-dimensional resilient consensus, we solve multiple instances of lower-dimensional resilient consensus problems. Since a network exhibits improved resilience in low-dimensional states, the overall resilience is improved, albeit with reduced accuracy. Our main contributions are:

\begin{itemize}
\item We discuss the notion of accuracy in the resilient vector consensus problem and propose a framework to study the relationship between accuracy and resilience against adversarial agents.

\item We formulate the resilient bounded consensus problem to analyze an interplay between resilience and accuracy in higher dimensions and propose an algorithm to solve it.

%\item We propose an algorithm for the resilient bounded consensus in which a $d$-dimensional resilient consensus problem is solved by executing multiple instances of lower-dimensional resilient consensus problems.%, where the dimension is smaller than $d$.

\item We analyze the algorithm and present various resilience and accuracy bounds that demonstrate how the resilience against adversarial agents improves at the cost of accuracy. We also numerically evaluate our results.

% \item We give a formal framework to analyze accuracy vs resilience based on how to compare two polytopes...

% \item We compute the $d$-dimensional convex polytope in which the 
\end{itemize}

The rest of the paper is organized as follows: Section \ref{sec:Preliminaries} presents preliminaries. Section \ref{sec:RDCHD} provides an overview of the resilient multi-dimensional consensus problem. Section \ref{sec:RBC} formulates the resilient bounded consensus problem to study the trade-off between resilience and accuracy and also presents an algorithm. Section \ref{sec:Analysis} analyzes the proposed algorithm. Section \ref{sec:NE} illustrates a numerical example, and Section \ref{sec:Con} concludes the paper. 
%------------------------- Extra --------------
% \cite{Mudassir2020ACC}
% \cite{park2017fault}
% \cite{vaidya2014iterative}
% \cite{leblanc2013resilient}
% \cite{mendes2015multidimensional}
%--------- What is left -------
% 1) Abstract
% 2) Reference Additions.
% 3) Contributions List.
% 4) Organization Para.
% 5) Refine + 2nd shot.
%========== Section: PRELIM =============
\section{Preliminaries}
\label{sec:Preliminaries}
We consider a network of agents modeled by a directed graph $\mathcal{G}=(\mathcal{V},\mathcal{E})$, where $\mathcal{V}$ represents agents and $\mathcal{E}$ represents interactions between agents. %Each agent $i\in\mathcal{V}$ has a $d$-dimensional state vector whose value is updated over time. 
The state of each agent $i\in\calV$ {at time $t$} is represented by {a point $x_i(t) \in \mathbb{R}^d$}. An edge $(j,i)$ means that $i$ can observe the state value of $j$. The \emph{neighborhood} of $i$ is the set of nodes $\mathcal{N}_i =~\{j\in~\mathcal{V}| (j,i)\in\mathcal{E}\}\cup\{i\}$. %, and the \emph{closed} neighborhood of $i$ is $\mathcal{N}_i\cup\{i\}$. 
For a given set of points ${X}\subset\mathbb{R}^d$, we denote its \emph{convex hull} by conv$({X})$. {A set of points in $\mathbb{R}^d$ is said to be in \emph{general positions} if no hyperplane of dimension $d-1$ or less contains more than $d$ points}. A point $x\in\mathbb{R}^d$ is an \emph{interior point} of a set ${X}\subset\mathbb{R}^d$ if there exists an open ball centered at $x$ which is completely contained in ${X}$. Let $X_1\subset\mathbb{R}^{d_1}$ and $X_2\subset\mathbb{R}^{d_2}$, then the \emph{Cartesian product} of their convex hulls, denoted by conv$({X_1})\times$ conv$({X_2})$ is $\{(x_1 \;\;x_2)|\;x_1\in\text{conv}(X_1)\; \text{and }x_2\in\text{conv}(X_2) \}$. We use terms agents and nodes interchangeably, and similarly use terms points and states interchangeably.  

\paragraph*{Normal and Adversarial Agents}There are two types of agents in the network, {normal} and {adversarial}. \emph{Normal} agents synchronously interact with their neighbors and update their states according to a pre-defined state update rule, which is the consensus algorithm. \emph{Adversarial} agents can change their states arbitrarily and do not follow the pre-defined update rule. Moreover, an adversarial agent can transmit different values to nodes in its neighborhood, referred to as the \emph{Byzantine} model. $ F_i$ denotes the number of adversarial agents in the neighborhood of agent $i$. A normal agent cannot distinguish between its normal and adversarial neighbors. %For a normal node $i$, all nodes in its neighborhood are indistinguishable, that is, $i$ cannot identify which of its neighbors are adversarial.

\paragraph*{Resilient Vector Consensus}
The goal of the resilient vector consensus is to ensure the following two conditions:
(1) \emph{Safety --} Let ${X}(0) = \{x_1(0),x_2(0),\cdots,x_n(0)\}\subset\mathbb{R}^d$ be the set of initial states of normal nodes, then at each time step $t$, and for any normal node $i$, $x_i(t) \in\text{conv}({X}(0))$. 

\noindent
(2) \emph{Agreement --} For every $\epsilon>0$, there exists some $t_\epsilon$, such that $||x_i(t) -~ x_j(t)|| < ~\epsilon $ for all $t > t_\epsilon$, and for all normal node pairs $i,j$.

% For every $t\ge 0$, normal node pair $i,j$, and $\epsilon > 0$, there exists some $t_\epsilon$, such that $||x_i(t) -~ x_j(t)|| < ~\epsilon $ for all $t > t_\epsilon$. 

%======= Section: RES CON IN D-DIM ======
\section{Resilient Distributed Consensus in $\mathbb{R}^d$}
\label{sec:RDCHD}
In this section, first, we briefly discuss a resilient distributed vector consensus algorithm, known as the \emph{Approximate Distributed Robust Convergence (ADRC)}, proposed by Park and Hutchinson \cite{park2017fault}. Second, we discuss a computational improvement in the algorithm discussed in~\cite{Mudassir2020ACC}.% that improves the resilience of ADRC.

%---- Subsection: ADRC
\subsection{Approximate Distributed Robust Convergence (ADRC)}
The ADRC algorithm guarantees the consensus of normal agents in $\mathbb{R}^d$ if the number of adversarial agents in the neighborhood of each normal agent is bounded by a certain value that depends on $d$. The notion of \emph{$F$-safe point} is crucial to understanding the algorithm.

\begin{definition} (\emph{$F$-safe point})
\label{def:safe}
Given a set of $N$ points in $\mathbb{R}^d$, of which at most $F$ are adversarial, then a point $p$ that is guaranteed to lie in the interior of the convex hull of $(N-F)$ normal points is an $F$-safe point.
\end{definition}

%Note that the identities of $F$ adversarial points in a set of $N$ points are not known.

The ADRC algorithm relies on the computation of an $F_i$-safe point by every normal agent $i$ having $N_i$ agents in its neighborhood, of which at most $F_i$ are adversaries. %Note that each agent corresponds to a point in $\mathcal{R}^d$. 
The ADRC is a synchronous iterative algorithm in which each normal agent $i$ updates its state as follows \cite{park2017fault}:
\begin{itemize}
    \item In the iteration $t$, a normal agent $i$ gathers the state values of its neighbors $\mathcal{N}_i(t)$.
    \item Then, it computes an $F_i$-safe point, denoted by $s_i(t)$, of points corresponding to its neighbors' states.
    \item Agent $i$ then updates it's state as below.
    \begin{equation}
        \label{eq:ADRC}
        x_i(t+1) = \alpha_i(t)s_i(t) + (1 - \alpha_i(t))x_i(t),
    \end{equation}
    where $\alpha_i(t)$ is a dynamically chosen parameter in the range (0  1), whose value depends on the application \cite{park2017fault}.
\end{itemize}

If all normal agents follow the above routine, they are guaranteed to converge at some point in the convex hull of their initial states \cite{park2017fault}. The biggest challenge here is to ensure that each normal agent can compute a safe point. %So, the the main question is under what conditions a normal node $i$ can compute an $F_i$-safe point?
For this, \cite{park2017fault} utilized the ideas from Discrete Geometry and presented the following result.

\begin{prop} \emph{(Theoretical bound)}
\label{prop:safe_condition}
Given a set of $N$ points in general positions in $\mathbb{R}^d$, where $d\in\{2,3,\cdots, 8\}$, and at most $F$ points are adversarial, then it is possible to find an $F$-safe point if 
\begin{equation}
\label{eqn:theory}
N \ge (F+1)(d+1).
%F \le \left\lceil\frac{N}{d+1}\right\rceil - 1.
\end{equation}
\end{prop}

In particular, \cite{park2017fault} used the notion of \emph{Tverberg partition} to compute an $F$-safe point. The main idea is to partition a set of $N$ points in $\mathbb{R}^d$ into $(F+1)$ parts such that the convex hull of points in one part has a non-empty intersection with the convex hull of points in any other part. An $F$-safe point is an interior point in the intersection of these $F+1$ convex hulls. In general, the computation of Tverberg partition of points is an NP-hard problem. The best known approximation algorithm runs in $d^{O(1)N}$ time and computes a Tverberg partition of $N$ points into $F+1$ parts if $F \le \lceil\frac{N}{2^d}\rceil - 1$. In other words, we can state the following:

\begin{prop} \emph{(Practical bound)}
Given a set of $N$ points in general positions in $\mathbb{R}^d$, of which at most $F$ are adversarial and $d\le 8$, then it is possible to compute an $F$-safe point (using Tverberg partition) if
\label{prop:practical}
\begin{equation}
    \label{eq:practical}
    F \le \left\lceil\frac{N}{2^d}\right\rceil - 1.
\end{equation}
\end{prop}

Thus, \eqref{eqn:theory} and \eqref{eq:practical} provide theoretical and practical resilience guarantees of the ADRC algorithm, respectively. Consequently, in a network $\mathcal{G}$, if a normal agent $i$ has $N_i$ neighbors and at most $F_i$ of them are adversarial, then resilient consensus is guaranteed by the ADRC algorithm, if $F_i \le \lceil\frac{N_i}{2^d}\rceil - 1$, for every normal agent $i$.

%---- Subsection: Centerpoint Cons. -----
\subsection{Resilient Vector Consensus Using Centerpoint}
Recently, \cite{Mudassir2020ACC} proposed to utilize the notion of \emph{centerpoint} instead of Tverberg partition to compute a safe point. Centerpoint can be viewed as an extension of the median in higher dimensions and is defined below.

\begin{definition}
\label{def:centerpoint}
\emph{(Centerpoint)}
Given a set $X$ of $N$ points in $\mathbb{R}^d$ in general positions, a centerpoint $p$ is a point, not necessarily from $X$, such that any closed half-space\footnote{A closed half-space in $\mathbb{R}^d$ is a set of the form $\{x\in\mathbb{R}^d: a^Tx \ge b\}$ for some $a\in \mathbb{R}^d\setminus \{0\}$.} of $\mathbb{R}^d$ containing $p$ also contains at least $\frac{N}{d+1}$ points from $X$.
\end{definition}

It is shown in \cite{Mudassir2020ACC} that for a given set of $N$ points in $\mathbb{R}^d$, an $F$-safe point is essentially an interior centerpoint for~$F=\frac{N}{d+1}-1$. For instance, consider six points in $\mathbb{R}^2$ in Figure \ref{fig:CE}(a). The gray region is the set of all centerpoints. At the same time, the gray region is also the set of all $1$-safe points. It means that no matter which one of the six points is adversarial, every point in the the gray region is guaranteed to lie in the convex hull of remaining five normal points. For example, in Figure \ref{fig:CE}(b), the red point is adversary and the yellow region is the convex hull of normal points. We observe that all $1$-safe points (gray region) lie inside the convex hull of normal points. The same is illustrated for a different adversary point in Figure \ref{fig:CE}(c). %Thus, the region of $1$-safe points is same as the region of centerpoints. 

\begin{figure}[htb]
\centering
\begin{subfigure}[b]{0.14\textwidth}
\centering
\includegraphics[scale=0.5]{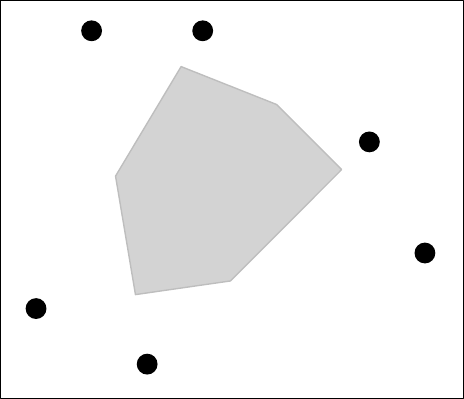}
\caption{}
\end{subfigure}
\begin{subfigure}[b]{0.14\textwidth}
\centering
\includegraphics[scale=0.5]{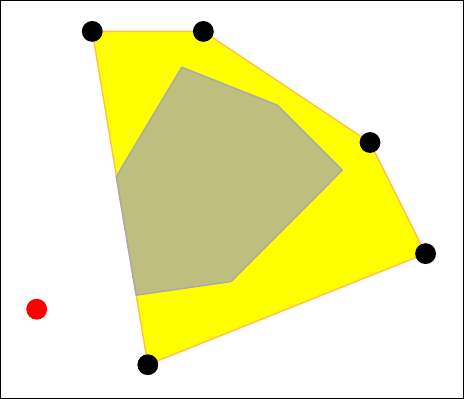}
\caption{}
\end{subfigure}
\begin{subfigure}[b]{0.14\textwidth}
\centering
\includegraphics[scale=0.5]{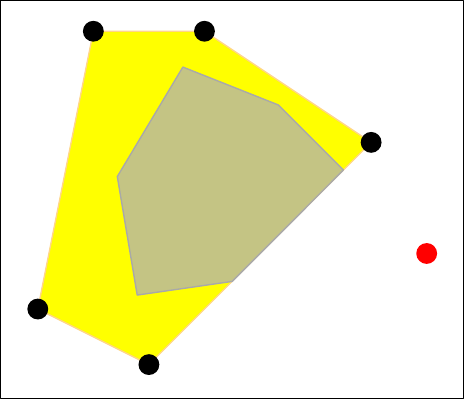}
\caption{}
\end{subfigure}
\caption{The region of centerpoints (gray area) is same as the region of $1$-safe points.}
\label{fig:CE}
\end{figure}

Consequently, instead of using Tverberg partition, we can use centerpoint to compute an $F$-safe point. 
{We also note that every point that lies in the intersection of the Tverberg partition of $N$ points in $(F+1)$ parts, where $F=\lceil\frac{N}{d+1}\rceil-1$, is also a centerpoint (and hence, an $F$-safe point). However, the converse is not true in general, that is, a point in the intersection of Tverberg partition of points may not be a centerpoint.}
The main advantage of using centerpoint is that it improves the practical resilience guarantees of the ADRC algorithm. In particular, we have the following result \cite{Mudassir2020ACC}.

\begin{prop}
\label{prop:practical_CP_1}
Given a set of $N$ points in general positions in $\mathbb{R}^d$, of which at most $F$ are adversarial, then an $F$-safe point can be computed (using centerpoint) if

\begin{equation}
\label{eq:practical_CP}
\begin{split}
 F & \le \left\lceil\frac{N}{d+1}\right\rceil - 1, \hspace{0.2in} \text{for } d = 2,3, \text{ and}\\
 F & = \Omega\left(\frac{N_i}{d^2}\right)  \hspace{0.7in} \text{for } d > 3,
\end{split}
\end{equation}
where $r>1$ is some positive integer. Moreover, such an $F$-safe point can be computed in $O(N)$ and $O(N^2)$ times in $\mathbb{R}^2$ and $\mathbb{R}^3$ respectively, and in $O\left( N^{c\log d} (2d)^d \right)$ in $\mathbb{R}^d$ for $d>3$, where $c$ is some constant.
\end{prop}

%Comparing \eqref{eq:practical} and \eqref{eq:practical_CP}, we find that centerpoint improves resilience (practically) of the ADRC algorithm.
{For algorithmic details of computing a centerpoint, we refer readers to \cite{jadhav1994computing,chan2004optimal,miller2010approximate}.}

%======== Section: New Section ==========
\section{Resilience-Accuracy Trade-off and\\ Resilient Bounded Consensus}
\label{sec:RBC}
If the number of adversaries $F_i$ in the neighborhood of a normal agent $i$ satisfies \eqref{eqn:theory}, then all normal agents are guaranteed to converge in the convex hull of their initial points. Here, we are interested in analyzing the interplay between between resilience and accuracy of the algorithm. In other words, \emph{what are the implications if the number of adversaries is greater than the one in \eqref{eqn:theory}?} Can the normal agents still converge? If they do, how far could the agreement point be from the convex hull of initial points?

First, we note that if $F\ge \lceil N/(d+1)\rceil$, then an $F$-safe point does not exist. To illustrate this, consider an example in Figure \ref{fig:AC} with $N=6$ points and $F=2$. Assuming the right-most two points are adversarial, the convex hull of normal nodes is the yellow region in Figure \ref{fig:AC}(a), and a $2$-safe point must lie in the interior of this region. Similarly, consider another situation in which the left-most two nodes are adversarial, as shown in Figure \ref{fig:AC}(b). Then, a $2$-safe point must lie in the interior of the convex hull of normal nodes shown as the green region. We observe that the intersection of interiors of convex hulls of normal points in Figures \ref{fig:AC}(a) and (b), shown as yellow and green regions respectively, is an empty set, which means there is no $2$-safe point here. Thus, $F \le \lceil\frac{N}{d+1}\rceil - 1$ is not only a sufficient but also a necessary condition for the existence of an $F$-safe point.   

\begin{figure}[htb]
\centering
\begin{subfigure}[b]{0.14\textwidth}
\centering
\includegraphics[scale=0.5]{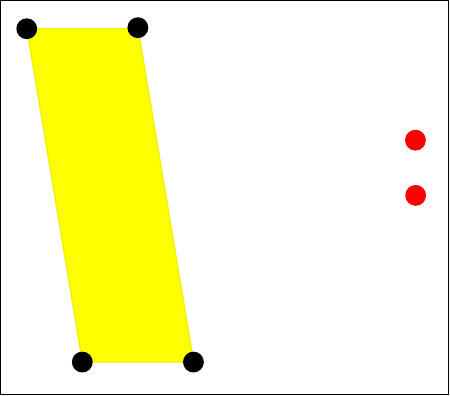}
\caption{}
\end{subfigure}
\begin{subfigure}[b]{0.14\textwidth}
\centering
\includegraphics[scale=0.5]{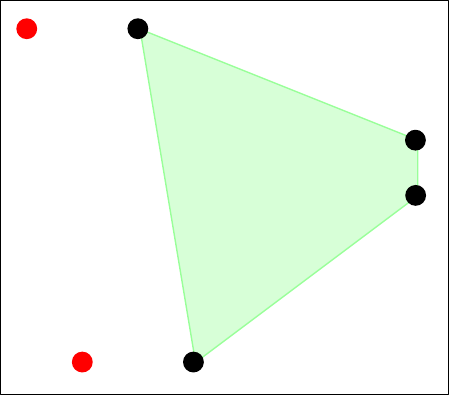}
\caption{}
\end{subfigure}
\begin{subfigure}[b]{0.14\textwidth}
\centering
\includegraphics[scale=0.5]{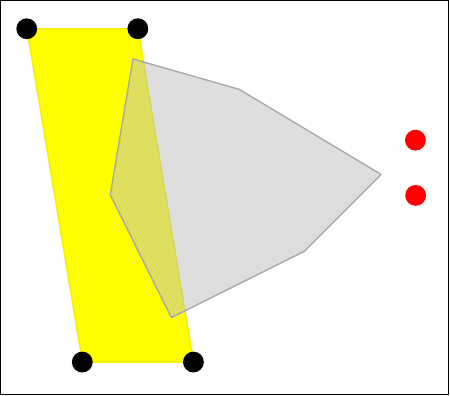}
\caption{}
\end{subfigure}
\caption{The interior of convex hulls of normal nodes in (a) and (b) do not intersect. In (c), the gray region is the centerpoint region of all six points.}
\label{fig:AC}
\end{figure}

Second, we note that if $F\ge \lceil \frac{N}{d+1}\rceil$, then the centerpoint of a cloud of $N$ points can be arbitrarily far away from the convex hull of normal points.\footnote{By the Centerpoint Theorem, every finite set of points in $\mathbb{R}^d$ has some centerpoint \cite{matousek2013lectures}.} For instance, in Figure~\ref{fig:AC}(c), every point in the gray region is a centerpoint. The adversaries (red points) can move arbitrarily far away from the normal points, and there would be centerpoints arbitrarily far away from the convex hull of normal points. Thus, in the ADRC algorithm, if a normal node $i$ has $N_i$ neighbors, of which $F_i\ge \lceil\frac{N_i}{d+1}\rceil$ are adversarial, and $i$ updates its state based on a centerpoint of its neighbors' states that is chosen arbitrarily from the centerpoint region, then $x_i(t)$ can be arbitrarily far away from the convex hull of normal nodes' initial states.
% We can also mention that if ADRC computes a Tverberg point, this can still happen as every ceneterpoint is a Tverberg point and therefore there are Tverberg points that are arbitrarily far away.

% \begin{figure}
%     \centering
%     \includegraphics[scale=0.85]{CAPD.pdf}
%     \caption{Caption}
%     \label{fig:CPAD}
% \end{figure}

% Next, we ask if it is possible to guarantee the convergence of normal nodes in some bounded region $\mathcal{B}$ in case $F\le F'$, where $F' \ge \frac{N}{d+1}$. 

Next, we ask \emph{if it is possible to guarantee the convergence of normal agents in some bounded region $\mathcal{B}$ if each normal agent $i$ satisfies $\lceil\frac{N_i}{d+1}\rceil\le F_i\le F_i'$ for some $F_i'$?} Here, we can expect $\mathcal{C}\subseteq \mathcal{B} \subset \mathbb{R}^d$, where $\mathcal{C} = \text{conv}(X(0))$ is the convex hull of the points corresponding to normal agents' initial states. In other words, is it possible to improve the resilience of the consensus algorithm in $d$ dimensions at the cost of accuracy, where accuracy measures how far away from $\mathcal{C}$ do normal agents converge? To formalize this, we define a \emph{Resilient Bounded Consensus} problem below.

\begin{definition} (\emph{Resilient Bounded Consensus})
\label{def:RACC}
Consider a network $\mathcal{G}(t)=(\mathcal{V},\mathcal{E}(t))$ of agents in which each normal agent $i$ has $N_i(t)$ neighbors, of which at most $F_i(t)$ are adversarial. Normal agents in $\calV$ update their states such that at each time step $t$ and for every normal agent $i$, the state $x_i(t)$ is in a bounded convex region $\mathcal{B}\subset\mathbb{R}^d$ (irrespective of the states of adversarial agents). Moreover, for every $\epsilon>0$, there exists some $t_\epsilon$ such that $||x_i(t) -~ x_j(t)|| < ~\epsilon $ for all $t > t_\epsilon$ and for all normal node pairs $i,j$.
\end{definition}

Here, we consider $\mathcal{B}$ such that $\calC\subseteq\calB$. If $\mathcal{B}=\mathcal{C}$, we get the typical resilient vector consensus problem (Section \ref{sec:Preliminaries}). 

%\vspace{-0.1in}
\subsection{Accuracy and Resilience in Resilient Bounded Consensus} If $x^\ast\in\calB$ is a consensus point of all normal agents in the resilient bounded consensus, then the distance between $x^\ast$ and $\calC$, denoted by $\delta(x^\ast,\calC)$, is defined as,
\begin{equation}
\label{eq:xC}
\delta(x^\ast,\calC) =  \min_{c\in\partial\mathcal{C}}|| x^\ast - c||,
\end{equation}
where $\partial\mathcal{C}$ is the boundary of $\mathcal{C}$, and $|| x^\ast - c||$ is the Euclidean distance between points $x^\ast$ and $c$. To quantify the \emph{accuracy} of bounded consensus---how far can the agreement point in $\mathcal{B}$ be from $\mathcal{C}$---we use the notion of \emph{Hausdorff distance}, which is often used to measure how well one convex shape approximates the other \cite{bronstein2008approximation,lopez2005hausdorff}.

\begin{definition} (\emph{Hausdorff Distance})
\label{def:Hausdorff}
Given two convex regions $\mathcal{B}, \mathcal{C}\subset \mathbb{R}^d$, the \emph{Hausdorff distance} from $\mathcal{B}$ to $\mathcal{C}$ is %, denoted by $\delta(\mathcal{B},\mathcal{C})$, is the maximum Euclidean distance between a point in the boundary of $\mathcal{B}$ and the boundary of $\mathcal{C}$, that is,
\begin{equation}
\label{eq:HD}
\delta(\mathcal{B},\mathcal{C}) = \max_{b\in\partial\mathcal{B}}\min_{c\in\partial\mathcal{C}}|| b - c||,
\end{equation}
where $|| b - c||$ is the Euclidean distance between $b$ and $c$, and $\partial\mathcal{B}$, $\partial\mathcal{C}$ are the boundaries of $\mathcal{B}$ and $\mathcal{C}$, respectively.
%Normally, approximation error is defined relative to $X$’s diameter, where $X$ is the actual convex hull.
\end{definition}

%Given two convex bodies $X, Y\subset \mathbb{R}^d$ \emph{Hausdorff distance} is the maximum Euclidean distance between a point in the boundary of $X$ or $Y$ and the boundary of the other body. 
Note that $\delta(x^\ast,\calC) \le \delta(\mathcal{B},\mathcal{C})$. Typically, we state the accuracy of resilient bounded consensus relative to the \emph{diameter} of $\calC$, which is denoted by $\mu(\calC)$, and defined as, 

\begin{equation}
    \label{eq:diamC}
    \mu(C) = \max_{c_1,c_2\in\mathcal{C}}|| c_1 - c_2||.
\end{equation}

We are interested in the ratio $\delta(\calB,\calC)/\mu(\calC)$ to examine the \emph{accuracy} of resilient bounded consensus.
% %--------- Volume Ratio Begins ---------
% Another way to observe the accuracy of resilient bounded consensus is by computing the ratio of volumes of $\mathcal{B}$ and $\mathcal{C}$,
% \begin{equation}
% \label{eq:vol}
% v(\mathcal{B},\mathcal{C}) = \frac{\text{Vol}(\mathcal{B})}{\text{Vol}(\mathcal{C})},
% \end{equation}
% where $\text{Vol}(\calB)$ and $\text{Vol}(\calC)$ are the volumes of $\calB$ and $\calC$ respectively. The volume ratio criteria is often used in shape approximation algorithms in $\mathbb{R}^d$ \cite{bronstein2008approximation,barequet2001efficiently}. %in which the objective is to compute, approximate or measure the smallest volume convex shape of a certain structure containing a given convex body, for instance \cite{bronstein2008approximation,barequet2001efficiently,zhou2002algorithms}.%\footnote{To measure how one convex shape (polytope) approximates the other, there are other approaches and metrics that are suited for different applications, for instance \cite{bronstein2008approximation,lopez2005hausdorff,arya2017combinatorial,bhaskara2019approximating}.}

% % \textcolor{red}{Normally, approximation error is defined relative to $X$’s diameter, where $X$ is the actual convex hull. There are other metric, for instance \cite{bronstein2008approximation,florian1989metric,lopez2005hausdorff,arya2017combinatorial,bhaskara2019approximating,blum2016sparse}. \cite{graham2017approximate,leung1997neural}}
% %----------- Volume Ratio Ends ------
The \emph{resilience} of the resilient bounded consensus algorithm is measured by the maximum number of adversarial agents $F_i$ in the neighborhood of a normal agent $i$, such that despite the presence of these adversarial agents, all normal agents achieve resilient bounded consensus %by converging to some common point 
inside the convex region $\calB$. If $\calB = \calC$, the resilience bound is $F_i \le \frac{N_i}{d+1}-1$. As $\calB$ grows ($\calB\supset\calC$), the accuracy deteriorates as $\delta(\calB,\calC)$ increases. At the same time, the resilience bound may improve. %, and convergence of normal nodes in $\calB$ is guaranteed even with more than $\frac{N_i}{d+1}-1$ adversaries in $\calN_i$.
We are interested in a resilient bounded consensus algorithm that exploits this \emph{resilience-accuracy} trade-off.

% \textcolor{blue}{(Describe the resilience of the algorithm.)}

% \textcolor{blue}{(Describe again why we introduced this problem, hoping resilience will improve and study how accuracy changes with improvement in resilience.)}

%========== Section:  ===================
\subsection{Resilient Bounded Consensus Algorithm}
Our approach to achieving resilient bounded consensus is to partition the $d$-dimensional state into parts, implement the centerpoint based resilient consensus algorithm (in lower dimensions) on each part, and then combine the results to get the updated $d$-dimensional state. The performance, in terms of resilience and accuracy, will depend on the partition of state to lower dimensions. We illustrate the approach by an example, and then outline the details.

Consider a complete network of seven agents, of which two are adversarial, and each agent's state is in $\mathbb{R}^3$. The initial positions of normal agents are $(3 \; 10 \; 10)$, $(5 \; 5 \; 5)$, $(10 \; 4 \; 2)$, $(1 \; 4 \; 4)$, $(4 \; 2 \; 9)$, and the convex hull of their initial positions $\mathcal{C}$ is shown in Figure~\ref{fig:example_1}(a). To guarantee the convergence of normal agents in $\mathcal{C}$, each normal agent needs to have at least nine neighbors according to \eqref{eqn:theory}, which is not the case here. Alternatively, we aim for a resilient bounded consensus in which the goal is to guarantee that all normal agents converge at some point in a bounded region $\mathcal{B}_1$ (Figure \ref{fig:example_1}(b)). Each normal agent implements two instances of resilient consensus algorithm (ADRC): a 2-dimensional resilient consensus on the first two coordinates of its neighbors' states, and a scalar consensus on the remaining third coordinate. Since $F_i=2$ and each normal node $i$ has seven neighbors, 2-dimensional resilient consensus ensures that the first two state coordinates of all normal nodes converge to a point in the convex hull of their initial values in those coordinates. Similarly, resilient scalar consensus guarantees that the third coordinate of the state of all normal nodes converges to a value in the range defined by their initial third coordinate values. Consequently, all normal nodes converge to a point in a polytope $\mathcal{B}_1$.
Here $\delta(\mathcal{B}_1,\mathcal{C}) = 5.4$, the diameter of $\mathcal{C}$ is 12.2, and hence $\delta(\mathcal{B}_1,\mathcal{C})/\mu(\mathcal{C}) = 0.44$. 
%Moreover, $v(\mathcal{B}_1,\mathcal{C}) = 4.1$. 
Thus, the resilience is improved at the cost of accuracy. In Figure~\ref{fig:example_1}(c), we illustrate the box $\mathcal{B}_2$ in which all normal agents converge to some point as a result of coordinate-wise resilient consensus. Here, %$v(\mathcal{B}_2,\mathcal{C}) = 8.2$,
$\delta(\mathcal{B}_2,\mathcal{C}) = 5.7$, and $\delta(\mathcal{B}_2,\mathcal{C})/\mu(\mathcal{C}) = 0.47$.

% \begin{figure}[ht]
% \centering
% \begin{subfigure}{.22\textwidth}
% \centering
% \includegraphics[scale=0.5]{K6m}
% \caption{$\mathcal{G}$}
% \end{subfigure}
% \begin{subfigure}{.2\textwidth}
% \centering
% \includegraphics[scale=0.24]{hull.png}
% \caption{$\mathcal{C}$}
% \end{subfigure}
% \begin{subfigure}{.22\textwidth}
% \centering
% \includegraphics[scale=0.24]{prism.png}
% \caption{$\mathcal{B}_1$}
% \end{subfigure}
% \begin{subfigure}{.2\textwidth}
% \centering
% \includegraphics[scale=0.24]{box.png}
% \caption{$\mathcal{B}_2$}
% \end{subfigure}
% \caption{Caption}
% \label{fig:example_1}
% \end{figure}

% %-----------
\begin{figure}[!h]
\centering
\begin{subfigure}{.155\textwidth}
\centering
\includegraphics[scale=0.21]{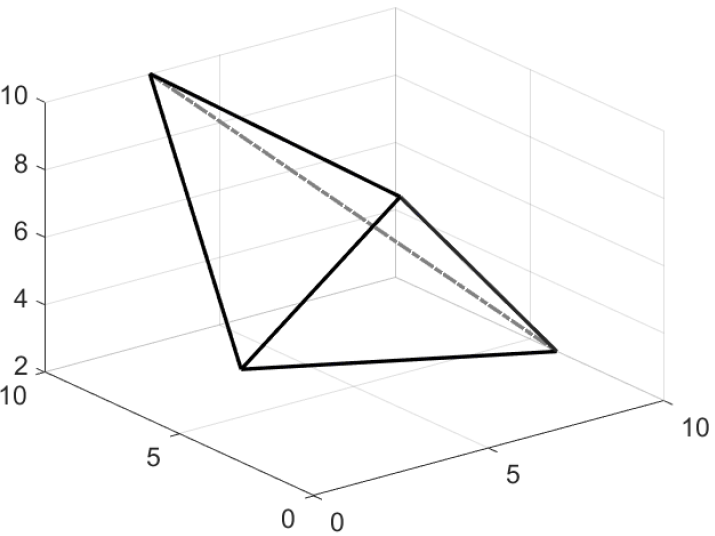}
\caption{$\mathcal{C}$}
\end{subfigure}
\begin{subfigure}{.155\textwidth}
\centering
\includegraphics[scale=0.21]{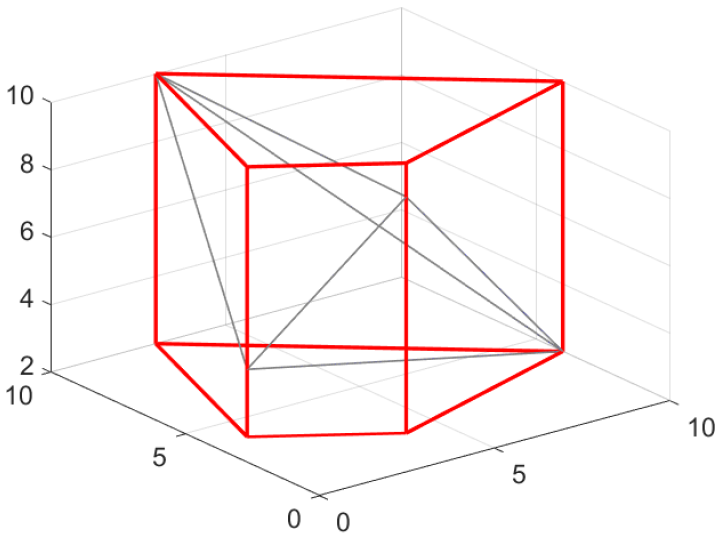}
\caption{$\mathcal{B}_1$}
\end{subfigure}
\begin{subfigure}{.155\textwidth}
\centering
\includegraphics[scale=0.21]{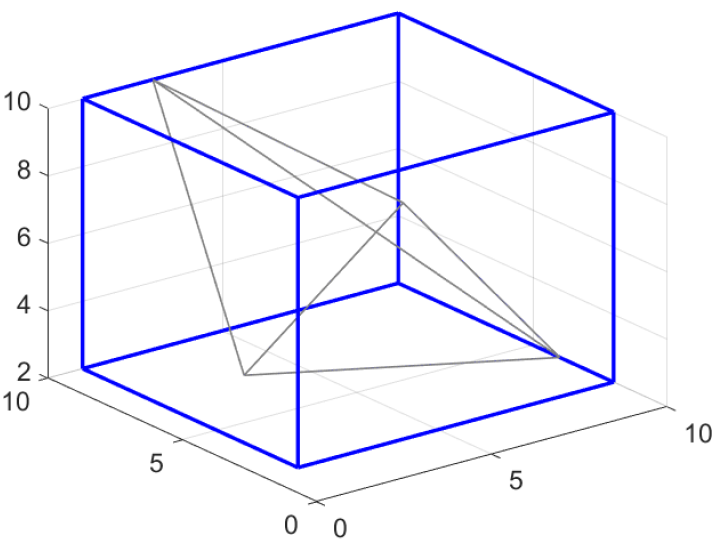}
\caption{$\mathcal{B}_2$}
\end{subfigure}
\caption{An example of resilient bounded consensus in $\mathbb{R}^3$.}
\label{fig:example_1}
\end{figure}
% %-------------

Next, we present \emph{Resilient Bounded Consensus (RBC)} algorithm for each normal agent $i$ in Algorithm \ref{algo:RBC}. First, we introduce some notations. Let $I = \{1,2,\cdots,d\}$, and $P = \{I_1,I_2,\cdots,I_k\}$ be a partition of $I $ into $k$ subsets. If $x_i(t) = [x_{i,j}(t)]_{j\in I}\in\mathbb{R}^d$, then $x_{i}^{\ell}(t)\in\mathbb{R}^{|I_\ell|}$ denotes the vector consisting of the values of $x_i(t)$ at the coordinates indexed in $I_\ell$, that is, $x_i^\ell(t) = [x_{i,j}(t)]_{j\in I_\ell}$. Further, we define $F_i^\ell(t)=\frac{N_i(t)}{|I_\ell|+1}-1$, where $N_i(t)=~|\calN_i(t)|$.% is the size of the neighborhood of agent $i$ at time $t$.

\begin{algorithm}
\caption{RBC for a Normal Agent $i$}
\label{algo:RBC}
\begin{algorithmic}[1]
%\footnotesize
\State \textbf{Given} Partition $P$ of $I=\{1,2,\cdots,d\}$.
\For {each iteration $t$}
\For{each $I_\ell \in P$} 
\State Compute $x_j^{\ell}(t)$, $\forall j\in\calN_i(t)$.
\State Compute an $F_i^\ell(t)$-safe point, say $s_i^\ell(t)$, by 
\Statex \hspace{.37in} computing a centerpoint of $\{x_j^{\ell}(t)\}$, $j\in\calN_i(t)$.
%\State Compute a centerpoint of points $x_j^{\ell}(t)$, $\forall j$. 
\State Update $x_i^{\ell}(t)$ by the following rule:
\Statex \hspace{0.6in}$x_i^{\ell}(t+1) = \alpha_i(t)s_i^{\ell}(t) + (1 - \alpha_i(t))x_i^{\ell}(t)$.
\EndFor
\State Combine $x_i^\ell(t+1)$, $\forall \ell \in \{1,\cdots,k\}$ to get the
\Statex \hspace{0.18in} updated state $x_i(t+1)\in\mathbb{R}^d$.
\EndFor
%\State \textbf{Return} $\calE'$
\end{algorithmic}
\end{algorithm}

In line 6, $\alpha_i(t)$ satisfies $0 < \alpha_i(t)< 1$, and is chosen depending on the specific application \cite{park2017fault}.
% %========================================
% \begin{algorithm}[H]
% \SetAlgoLined
% \KwResult{$x_i(t+1)$ }
%  \For{each iteration $t$}{
%     \For{each $I_\ell \in P$}{
%         Compute $x_j^{\ell}(t)$, $\forall j\in\calN_i$.\\
%         Compute an $F_i^\ell$-safe point, say $s_i^\ell(t)$, by computing a centerpoint of $x_j^{\ell}(t)$, $\forall j$.\\
%         Update $x_j^{\ell}(t)$ by the following rule:
% % $$
% % x_i^{\ell}(t+1) = \alpha_i(t)s_i^{\ell}(t) + (1 - \alpha_i(t))x_i^{\ell}(t),
% % $$
%     }
%  }
% %\caption{Resilient Bounded Consensus}
% \end{algorithm}
% %========================================
%========== Section:  ===================
% \section{Details}
% \begin{figure}[ht]
% \centering
% \begin{subfigure}{.15\textwidth}
% \centering
% \includegraphics[scale=0.33]{convexhull}
% \caption{}
% \end{subfigure}
% \begin{subfigure}{.15\textwidth}
% \centering
% \includegraphics[scale=0.33]{convexprism}
% \caption{}
% \end{subfigure}
% \begin{subfigure}{.15\textwidth}
% \centering
% \includegraphics[scale=0.33]{convexcube}
% \caption{}
% \end{subfigure}
% \caption{Caption}
% \label{fig:my_label}
% \end{figure}

%============ Section - ANALYSIS ========
\section{Analysis of the Resilient Bounded Consensus Algorithm}
\label{sec:Analysis}
In this section, we analyze the accuracy and resilience of Algorithm \ref{algo:RBC}.
%An immediate consequence of \cite[Theorem V.1]{park2017fault} is Theorem \ref{thm:main} below. 
First, we define the notions of \emph{jointly reachable} and \emph{repeatedly reachable} sequence of graphs \cite{park2017fault} needed to state the convergence of RBC algorithm. Let $\bar{\calG}(t) = (\bar{\calV}, \bar{\calE}(t))$ be a graph representing normal nodes $\bar{\calV}\subseteq\calV$ and edges between them at time $t$. 
\begin{definition} (\emph{Repeatedly reachable graph sequence})
\label{def:jR}
Let $j$ be a non-negative integer. A finite sequence of graphs $\bar{\calG}(T_j), \bar{\calG}(T_j+1) \cdots, \bar{\calG}(T_{j+1}-1) $, where each graph in the sequence has the same vertex set $\bar{\cal{V}}$ is called jointly reachable if the union of graphs defined as $$
\bigcup\limits_{t=T_j}^{T_{j+1}-1}\bar{\cal{G}}(t) = \left(\bar{\calV},\bigcup\limits_{t=T_j}^{T_{j+1}-1}\bar{\cal{E}}(t)\right)
$$
contains a vertex $v\in\bar{\calV}$ such that for every $v'\ne v$ there exists a path form $v'$ to $v$ in this union of graphs.
\end{definition}

\begin{definition} (\emph{Jointly reachable graph sequence})
\label{def:rR}
An infinite sequence of graphs $\bar{\calG}(0), \bar{\calG}(1) \cdots$ is called repeatedly reachable if there is a sequence of times $0=T_1 < T_2 < T_3 \cdots$ such that $T_{j+1} - T_j < \infty$ and the subsequence $\bar{\calG}(T_j), \bar{\calG}(T_j + 1),\cdots \bar{\calG}(T_{j+1} - 1)$ is jointly reachable $\forall j$.
\end{definition}
Basically, an infinite sequence $\bar{\calG}(0), \bar{\calG}(1) \cdots$ is repeatedly reachable if it can be partitioned into contiguous finite length subsequences that are themselves jointly reachable.
%------------------- Theorem Begins ---------------

Moreover, we define $X^\ell(0)$ to be the set of initial positions of normal agents at indices in $I_\ell$, that is, 
\begin{equation}
\label{eq:Xl}
X^\ell(0) := \{x_i^\ell(0)\}_{i \in \bar{\calV}},
\end{equation}
where $\bar{\calV}$ is the set of normal agents. Similarly, let $\calC^\ell$ be the convex hull of points in $X^\ell(0)$, that is,  
\begin{equation}
\label{eq:muCl}
\mathcal{C}^\ell := \text{Conv}(X^\ell(0)).
\end{equation}
A consequence of \cite[Theorem V.1]{park2017fault} is Theorem \ref{thm:main} below
\begin{theorem}
\label{thm:main}
Let $\mathcal{G}(t) = (\bar{\calV}\cup\calA,\calE(t))$ be a network of normal $\bar\calV$ and $\calA$ adversarial agents, where each $i\in(\bar\calV\cup\calA)$ has a state $x_i(t)\in \mathbb{R}^d$. Let $P = \{I_1,I_2,\cdots, I_k\}$ be a partition of $I = \{1,2,\cdots, d\}$ into $k$ subsets. Each $i\in\bar{\calV}$ implements Algorithm \ref{algo:RBC}, and has at most $F_i(t)$ adversaries in its neighborhood at time $t$. If 
\begin{equation}
    F_i(t) \le \left\lceil\frac{N_i(t)}{\max_\ell |I_\ell| + 1}\right\rceil - 1,
\end{equation}
 and the sequence of connectivity graphs of normal agents $\bar{\calG}(0), \bar{\calG}(1),  \cdots$ is repeatedly reachable, then all normal agents converge to a common point in $\calB$, which is a Cartesian product $\calC^1 \times \calC^2 \times \cdots \times \calC^k$.
\end{theorem}
%-------------- Theorem ends ----------------------
%--------------- Proof Begins ---------------------
\begin{proof}
In Algorithm \ref{algo:RBC}, each normal agent $i$ implements $k$ instances of ADRC algorithm. In the $\ell^{th}$ instance, at each time step $t$, $i$ gathers $x_j^\ell(t)\in\mathbb{R}^{|I_\ell|}$ for all $j\in\calN_i(t)$, computes an $F_i^\ell(t)$-safe point and update its state by moving towards the safe point. Now the convergence of $x_i^\ell(t)$, $\forall i\in \bar{\calV}$ to a common point in $\calC^\ell$ is guaranteed if $F_i(t)\le F_i^\ell(t)=\lceil\frac{N_i(t)}{|I_\ell|+1}\rceil-1$ and the sequence of connectivity graphs of normal nodes is repeatedly reachable \cite[Theorem V.1]{park2017fault}. Since
$$
F_i(t) \le\left\lceil\frac{N_i(t)}{\max_\ell |I_\ell| + 1}\right\rceil - 1 \le \left\lceil\frac{N_i(t)}{|I_\ell| + 1}\right\rceil - 1 = F_i^{\ell}(t),
$$
$x_i^\ell(t)$ for all $i\in \bar{\calV}$ converge to some point in $\calC^\ell$. This is true for all $\ell \in \{1,2,\cdots, k\}$, so an immediate consequence is that $d$-dimensional state $x_i(t)\in\mathbb{R}^d$ of each normal node $i$ converges to a common point in the cross product $\calC^1\times\calC^2\times\cdots\times\calC^k$, which is the desired result.
\end{proof}
%---------------- Proof Ends ----------------------
% Theorem \ref{thm:main} provides a resilience bound for the Algorithm~\ref{algo:RBC} in which a normal agent partitions a $d$-dimensional state into multiple lower dimensional states. The maximum dimension of state among these lower dimensional states determine the resilience of the RBC algorithm. Figure~\ref{fig:F_Comp} compares the resilience of ADRC and RBC algorithms. 
Theorem \ref{thm:main} provides a resilience bound for the Algorithm~\ref{algo:RBC} and Figure~\ref{fig:F_Comp} illustrates it in terms of the resilience of the ADRC algorithm. We consider a network in which agents have $d$-dimensional states and a normal agent $i$, implementing ADRC algorithm, is resilient against at most $F_i$ adversaries in its neighborhood. In Figure \ref{fig:F_Comp}(a), we plot the resilience of RBC (in terms of $F_i$) as a function of the maximum dimension of state obtained after partitioning a $d$-dimensional state. We observe that resilience improves as the maximum dimension of state in the partition decreases. For instance, if a $d$-dimensional state is partitioned into two $d/2$-dimensional states, the resilience of RBC improves by a factor of $\frac{\frac{2N}{d+2}-1}{\frac{N}{d+1}-1}$ as compared to the resilience of ADRC algorithm. In Figure~\ref{fig:F_Comp}(b), we fix $F_i$ and plot the number of agents needed in the neighborhood of a normal agent $i$ for it to be resilient against $F_i$ adversaries. Here, $N_i$ is the number of agents needed in the neighborhood of $i$ in the case of ADRC algorithm in $d$-dimensions. Again, we note that as the maximum dimension of state in the partition decreases, the required number of agents in the neighborhood of $i$ decreases. In other words, same resilience can be achieved with reduced local connectivity.
%Again, we note that as the maximum dimension of state in the partition decreases, the number of agents required in the neighborhood of $i$ to achieve resilience against $F_i$ adversaries decreases. In other words, same resilience is achieved with reduced local connectivity.
\begin{figure}[htb]
\centering
\begin{subfigure}[b]{0.23\textwidth}
\centering
\includegraphics[scale=0.165]{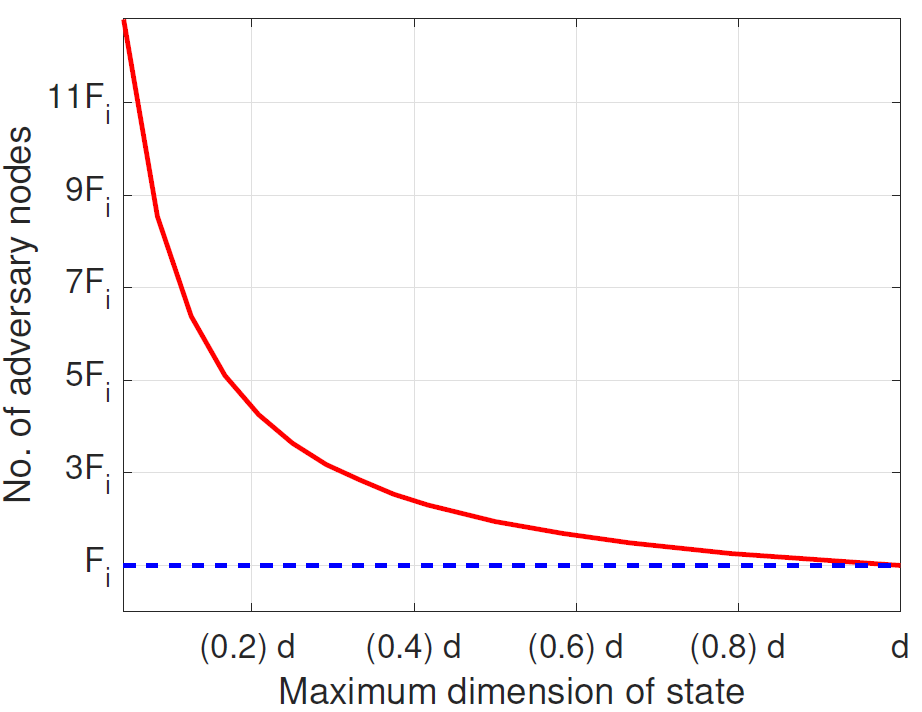}
\caption{Fixed $N_i$}
\end{subfigure}
\begin{subfigure}[b]{0.2\textwidth}
\centering
\includegraphics[scale=0.165]{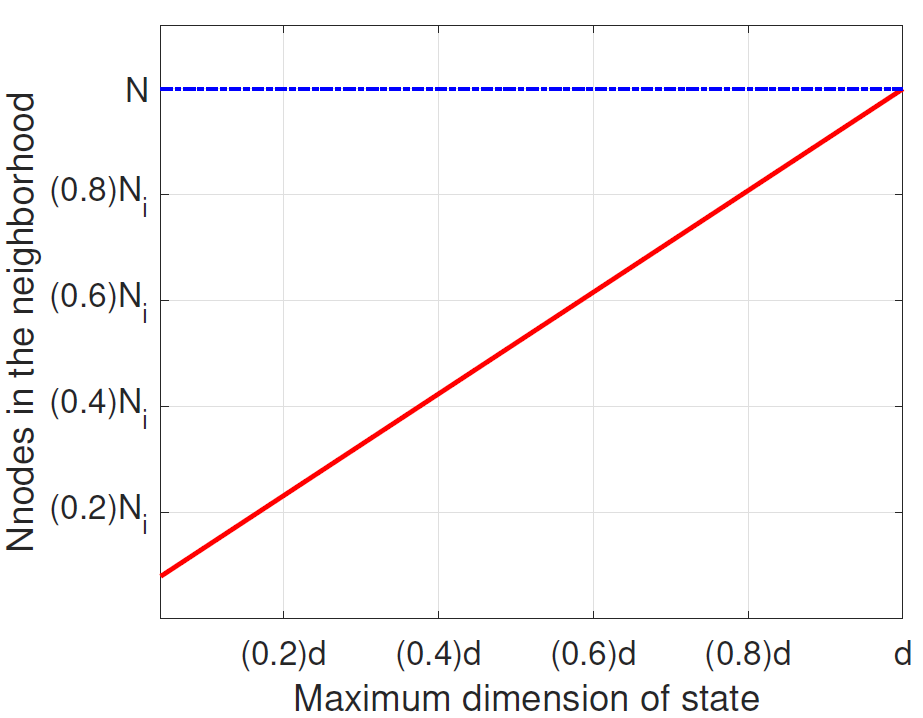}
\caption{Fixed $F_i$}
\end{subfigure}
\caption{Resilience of RBC as a function of the maximum dimension of state among partitioned states.}
\label{fig:F_Comp}
\end{figure}
%\textcolor{red}{Decomposition of states.}

\subsection{Accuracy of the Resilient Bounded Consensus Algorithm}
Next, we discuss the accuracy of RBC algorithm by computing the Hausdorff distance $\delta(\calB,\calC)$, where $\calB\supseteq\calC$ is the convex region in which normal nodes converge as a result of RBC. The shape of region $\calB$ and $\delta(\calB,\calC)$ depend on the partition of dimension. If we partition the $d$-dimensional state into $d$ 1-dimensional states (scalars) in RBC algorithm and each normal agent $i$ satisfies $N_i(t)\ge 2(F_i(t)+1)$, then all normal agents will converge in a hyperrectangle $\calB = \calC^1\times\calC^2\times\cdots\times\calC^d$, where $\calC^\ell$ is defined in \eqref{eq:muCl}. Note that each $\calC^\ell$ here is an interval in the $\ell^{th}$ dimension. We call such a convex region as the \emph{axis-parallel bounding box}, and note the following.

\begin{fact}
In RBC algorithm, $\delta(\calB,\calC)$ is maximum when $\calB$ is an axis-parallel bounding box.
\end{fact}
Thus, in the following our goal is to estimate the ratio $\delta(\calB,\calC)/\mu(\calC)$, where $\calC$ is the convex hull of a set of points in $\mathbb{R}^d$, $\calB$ is the corresponding axis-parallel bounding box, and $\mu(\calC)$ is the diameter of $\calC$. We start our discussion with the following conjecture:
\begin{conjecture}
For a given set of points in $\mathbb{R}^d$, let $\calC$ be the convex hull of points, $\mu(\calC)$ be the diameter of $\calC$ and $\calB$ be the corresponding axis-parallel bounding box, then
\begin{equation}
\label{eq:conjecture}
\delta(\calB,\calC) \le \sqrt{\frac{d}{2}}\;\mu(\calC).
\end{equation}
\end{conjecture}

We prove the above statement for $d=3$ and some other special cases. We believe that the above statement is true for any $d$, and would like to find a proof for the general case in the future.
%Although we do not have a proof for the general case, we strongly believe that the statement is true for any $d$. Here, we provide proofs for some special cases and also show that if true, the constant provided in the conjecture is the best possible. 
We begin by proving in three dimensions.

\begin{theorem}
\label{thm:3dproof}
Let $\calC$ be a convex hull of a given set of points in $\mathbb{R}^3$ and $\calB$ be the axis-parallel bounding box of $\calC$, then 
\begin{equation}
\label{eqn:Conjecture_3D}
\delta(\calB,\calC) \le \sqrt{\frac{3}{2}}\;\mu(\calC).
\end{equation}
\end{theorem}
%--------------theorem---------
%---------------proof-----------
\begin{proof} Without loss of generality, let all coordinates of input points be non-negative and let origin be the point on the bounding box that is at maximum distance from $\calC$. Let $a,b$ be the points in the given set that lie, respectively, on faces of bounding box parallel to planes $x=0$, and $y=0$ (illustrated in Figure \ref{fig:Thm_proof}). For the sake of contradiction, assume that $\min\{\norm{a},\norm{b}\}$ is more than $\sqrt{3/2 } \mu(\calC)$. Let $\pi(a), \pi(b)$ be the respective projections of $a$ and $b$, on the plane $z=0$. Then $\norm{\pi(a)}^2 \ge \norm{a}^2 - h_z^2$, and $\norm{\pi(b)}^2 \ge \norm{b}^2 - h_z^2$, where $h_z$ is the height of the bounding box $\calB$ in the $z$ direction. We have,
 $$
 \begin{array}{ccc}
 \mu(\calC) &\ge& \norm{a-b} \ge \norm{\pi(a) - \pi(b)} \\
		 &=& \sqrt{ \norm{\pi(a)}^2 + \norm{\pi(b)}^2 }\\
		 &\ge&  \sqrt{ \norm{a}^2 - h_z^2 + \norm{b}^2 - h_z^2 }\\
		 &\ge&  \sqrt{ 2\frac{3}{2} \mu(\calC)^2 - 2h_z^2   }\\
 \end{array}
 $$
 We have that $2 \mu(\calC)^2 - 2h_z^2 < 0$. This is clearly a contradiction since the height of a face can not be more than the diameter of the pointset. Thus, $\min \{\norm{a},\norm{b}\} \le \sqrt{ \frac{3}{2} }\times \mu(\calC)$. Since $\delta(\calB,\calC)\le \min \{\norm{a},\norm{b}\}$, we get the desired result.
 %We conclude that $\min (|a|,|b|)$, and thus the Hausdorff distance from the box to the convex hull, is at most $\sqrt{ \frac{3}{2} }\times \mu(\calC)$.
 \end{proof}
 
 \begin{figure}[h]
\centering
\includegraphics[scale=0.4]{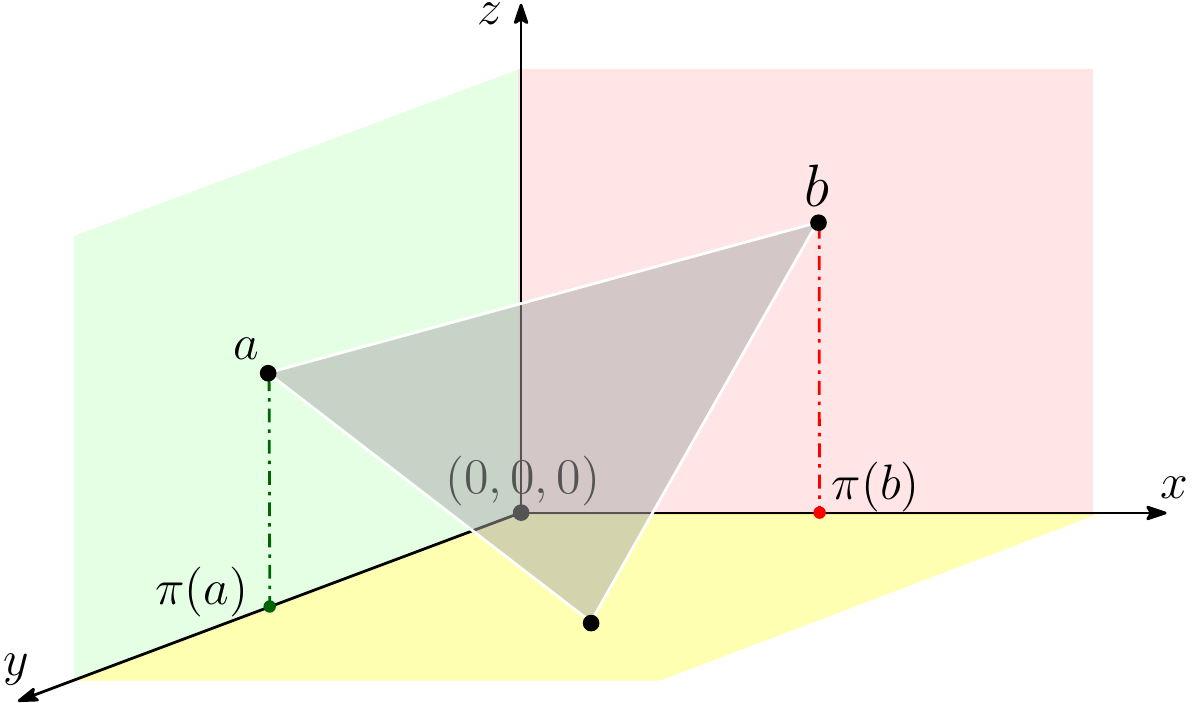}
\caption{Illustration of proof of Theorem \ref{thm:3dproof}.}
\label{fig:Thm_proof}
 \end{figure}
%--------------------------------------------------
In the following, we show that the statement of the conjecture is true for all $d$ for certain  \textit{symmetric} pointsets. Before we state and prove the result, we make some observations.

Since translation and rotation of points do not change the Hausdorff distance, we may assume that the origin is a point on $\calB$ with the maximum distance from $\calC$. As both $\calC$ and $\calB$ are convex, origin must be a corner vertex of the bounding box. For $\calB$ to be a minimum axis-parallel bounding box, each facets adjacent to origin must contain at least one point from the given pointset $P$. For $1\le i\le d$, let $p_i$ be a point in $P$ closest to origin for which the $i^{th}$ coordinate is zero. Then $p_1,p_2,\ldots,p_d$ define a hyperplane $X$. We observe the following:
\begin{fact}
The Hausdorff distance $\delta(\calB,\calC)$ is at most the Euclidean distance from origin to $X$ regardless of the positions of points $p_{d+1},p_{d+2},\ldots,p_n$.
\end{fact}
\begin{proof} We note that the points $p_1,p_2,\ldots,p_d$ that define a hyperplane $X$ lie on the boundary of $\calB$ thus, the closest point on $X$ can't be outside the box. Also, the points in $X\cap \calB$ are a convex combination of the points $p_1,p_2,\ldots,p_d$, and must lie in their convex hull, and therefore, must also lie in the convex hull $\calC$. The claim follows.
\end{proof}
Furthermore, there exist a pointset for which the Hausdorff distance $\delta(\calB,\calC)$ is equal to the Euclidean distance from origin to $X$, regardless of the positions of points $p_{d+1},p_{d+2},\ldots,p_n$.
This clearly is the case when all of the points $p_{d+1},p_{d+2},\ldots,p_n$ lie in the halfspace defined by $X$ that does not contain origin.

From these observations, we deduce that the upper bound on $\delta(\calB,\calC)$ is independent of the positions of points $p_{d+1},\ldots,p_n$. Thus, we only care about the first $d$ points.
%---------------------------------theorem------------------------------------------------
 \begin{theorem}
 \label{thm:symmetric}
	Let $\{p_1,p_2,\ldots,p_d \}$ be a subset of a pointset such that the $i^{th}$ coordinate of the point $p_i$  is zero and all other coordinates are set to some constant $a$. Then, 
	$$\delta(\calB,\calC) \le (d-1)/d\sqrt{\frac{{d}}{2}}\times \mu(\calC).$$
\end{theorem}
%-------------------------------end-theorem----------------------------------------------
%----------------------------------proof----------------------------------------------
\begin{proof}
If the constant $a$ is at most $\frac{1}{ \sqrt{2} }$, then the distance of each point from the origin is bounded by $\sqrt{  \frac{d}{2}  }$ and the claim follows. Therefore, we may assume that $a > \frac{1}{ \sqrt{2} }$. We further assume that $d>2$.
Clearly the diameter $\mu(\calC)$ in this case is at least $\norm{p_i-p_j}$, which is given by $\sqrt{2a^2}$. Consider the centroid point $z = \frac{1}{d}\sum_{i=1}^{d} p_i$, then% Then the norm of $z$ is given by,
\begin{equation*}
\begin{split}
\norm{z} & = \sqrt{   d \left( 1/d\times (d-1)a \right)^2  }\\
& =   (d-1)/d\times \sqrt{d a^2} \\
& = (d-1)/d \sqrt{\frac{{d}}{2}}\times  \mu(\calC),
\end{split}
\end{equation*}
which completes the proof. 
\end{proof}
We note in Theorem \ref{thm:symmetric}  that as $d$ goes to infinity, $(d-1)/d$ goes to one and this ratio goes to $\sqrt{\frac{{d}}{2}}\times  \mu(\calC)$.
%----------proof------------------------------------
The following proposition shows that the result in the previous theorem is best possible.
\begin{prop}
There exist pointsets for which $\delta(\calB,\calC) = (d-1)/d\sqrt{\frac{{d}}{2}}\times \mu(\calC)$.
\end{prop}
\begin{proof}
Consider a pointset that contain $2\times d$ points as follows: $\{p_1,p'_1,p_2,p'_2,\ldots,p_d,p'_d \}$ be such that the $i^{th}$ coordinate of the point $p_i$ is zero (similarly the $i^{th}$ coordinate of the point $p'_i$ is one) and all other coordinates are set to some constant $1/\sqrt{2}$. Clearly, the diameter of the pointset is one in this case. The closest point to the origin is the centroid $z$ of $p_1,p_2,\ldots,p_d$. We computed the norm of $z$ in the proof of Theorem \ref{thm:symmetric} to be
$$
\norm{z} =  (d-1)/d \sqrt{\frac{{d}}{2}}\times \mu(\calC).
$$
\end{proof}

We conclude this section with a following remark:
%--------- REMARK Begin ----
\begin{remark}
If the state dimension $d$ is $3$ and we partition it into a 2-dimensional and a scalar state, then as compared to \eqref{eqn:Conjecture_3D}, we can obtain a different bound on the Hausdorff distance from the bounded region $\calB$ (obtained as in Theorem \ref{thm:main}) to the convex hull of given set of points in $\mathbb{R}^3$. In particular, consider $I = \{1,2,3\}$ and a partition $P = \{I_1,I_2\}$ of $I$. Let $X = \{x_i\}$ be a set of points in $\mathbb{R}^3$, where $x_i = [x_{i,j}]_{j\in I}$. For each $x_i$, we partition it into $x_i^1$ and $x_i^2$ where $x_i^\ell = [x_{i,j}]_{j\in I_\ell}\in\mathbb{R}^{|I_\ell|}$. Further let $X^\ell = \{x_i^\ell\}$. % we obtain $x_i^1 = [x_{i,j}]_{j\in I_1}\in\mathbb{R}^{|I_1|}$ and similarly $x_i^2 = [x_{i,j}]_{j\in I_2}\in\mathbb{R}^{|I_2|}$. 
%Further, let $X^1 = \{x_i^1\}$ and $X^2 = \{x_i^2\}$. 
If $\calC, \calC^1, \calC^2$ are the convex hulls of points in $X, X^1, X^2$, respectively, and $\calB = \calC^1\times \calC^2$, then we can easily show that
\begin{equation}
\label{eq:other_bound}
\delta(\calB,\calC) \le \max \{\mu(\calC^1), \mu(\calC^2)\}.
\end{equation}
\end{remark}
\section{Numerical Illustration}
\label{sec:NE}

% %------------- Begin Figure (Graphs) ------------------
\begin{figure*}[ht]
\centering
\begin{subfigure}[b]{0.23\textwidth}
\centering
\includegraphics[scale=0.155]{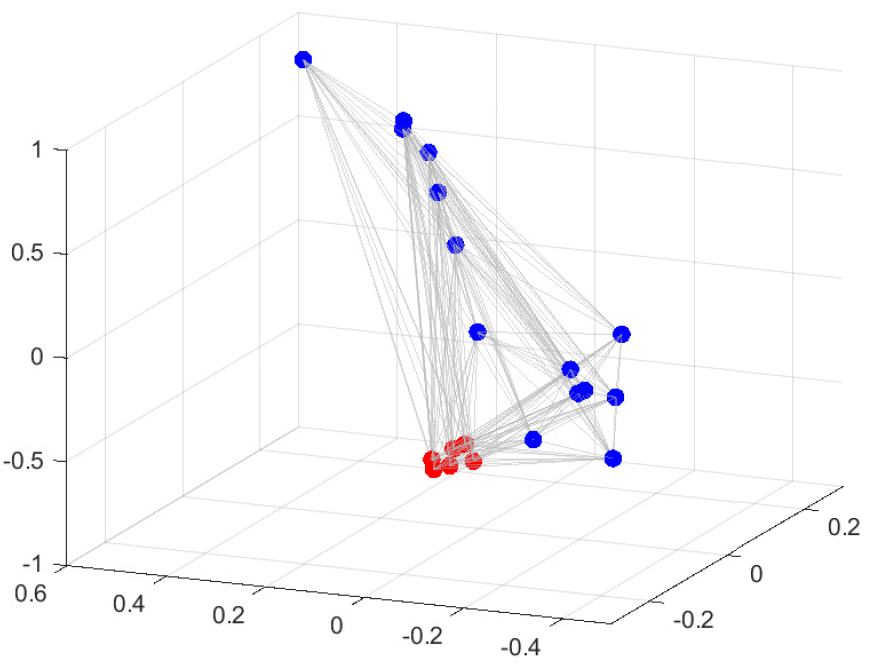}
\caption{}
\end{subfigure}\;
\begin{subfigure}[b]{0.21\textwidth}
\centering
\includegraphics[scale=0.143]{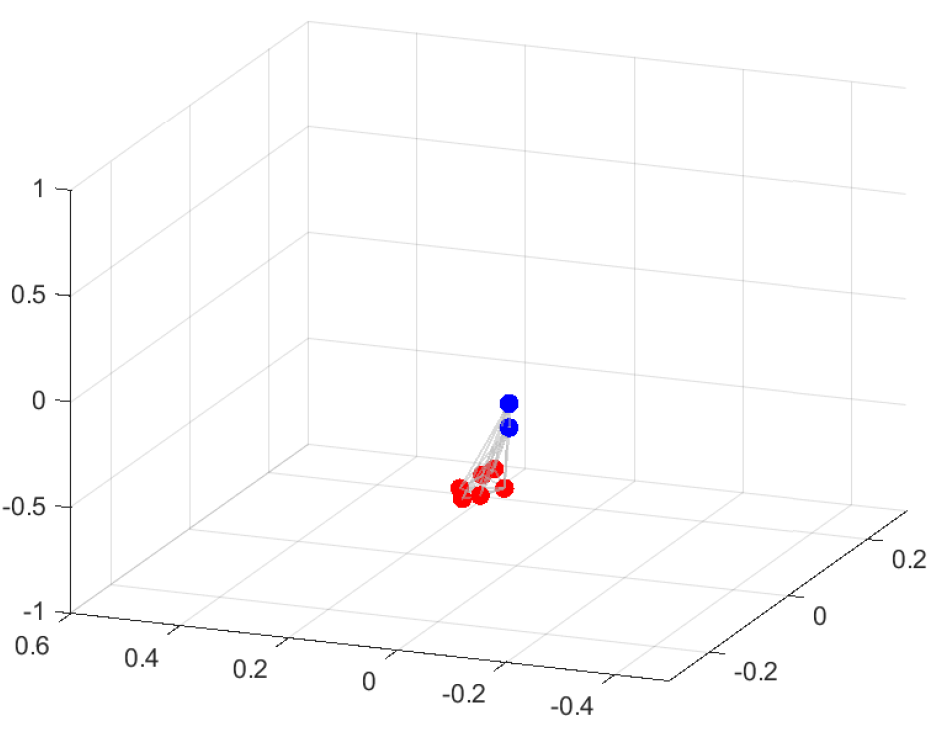}
\caption{}
\end{subfigure}
\begin{subfigure}[b]{0.21\textwidth}
\centering
\includegraphics[scale=0.143]{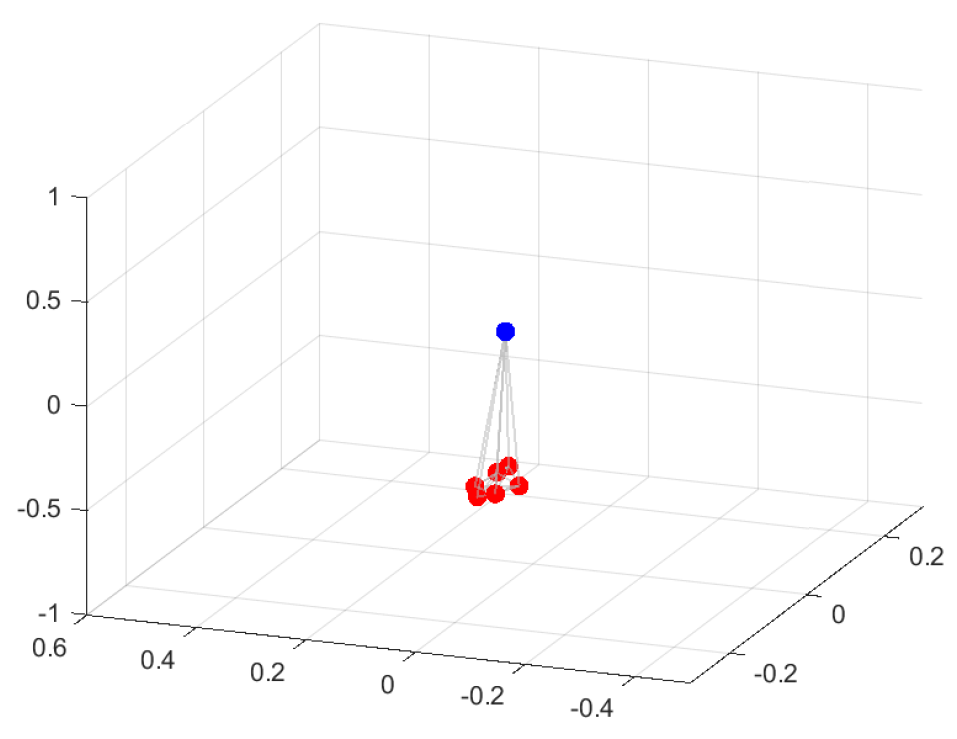}
\caption{}
\end{subfigure}
\begin{subfigure}[b]{0.21\textwidth}
\centering
\includegraphics[scale=0.143]{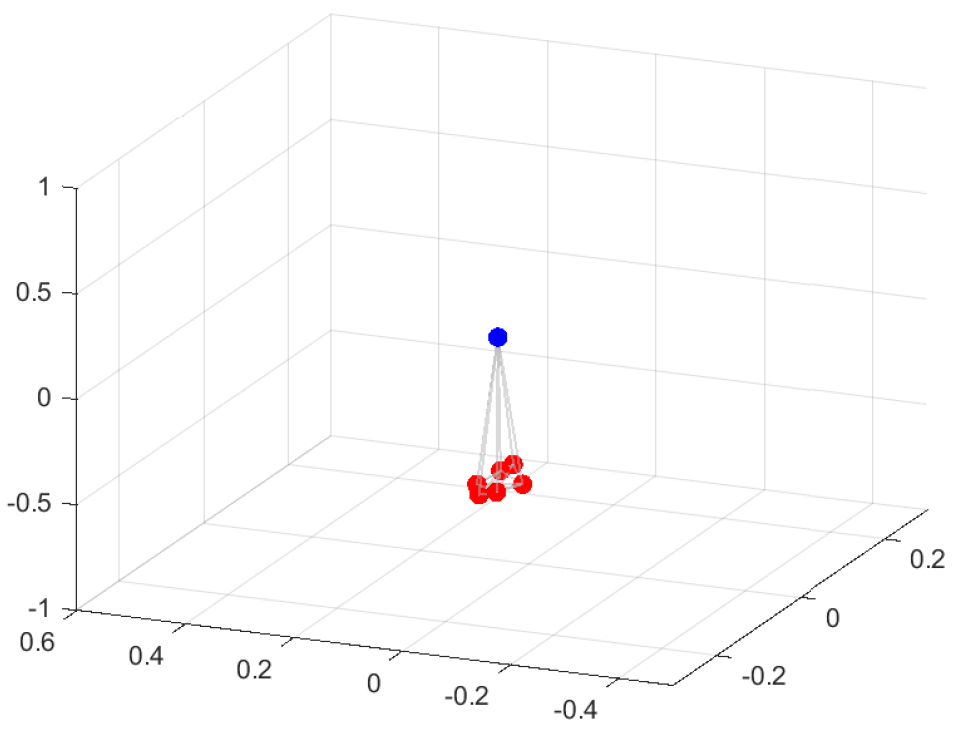}
\caption{}
\end{subfigure}
\caption{(a) Initial positions of agents. Final positions of agents as a result of (b) 3-dimensional ADRC, (c) RBC by partitioning 3-dimensional state into 2- and 1-dimensional states, and (d) coordinate-wise RBC.}
\label{fig:Num_Eval_Graphs}
\end{figure*}
%------------- End Figure (Graphs) --------------------

% %------------- Begin Figure (Hulls) ------------------
\begin{figure*}[htb]
\centering
\begin{subfigure}[b]{0.25\textwidth}
\centering
\includegraphics[scale=0.17]{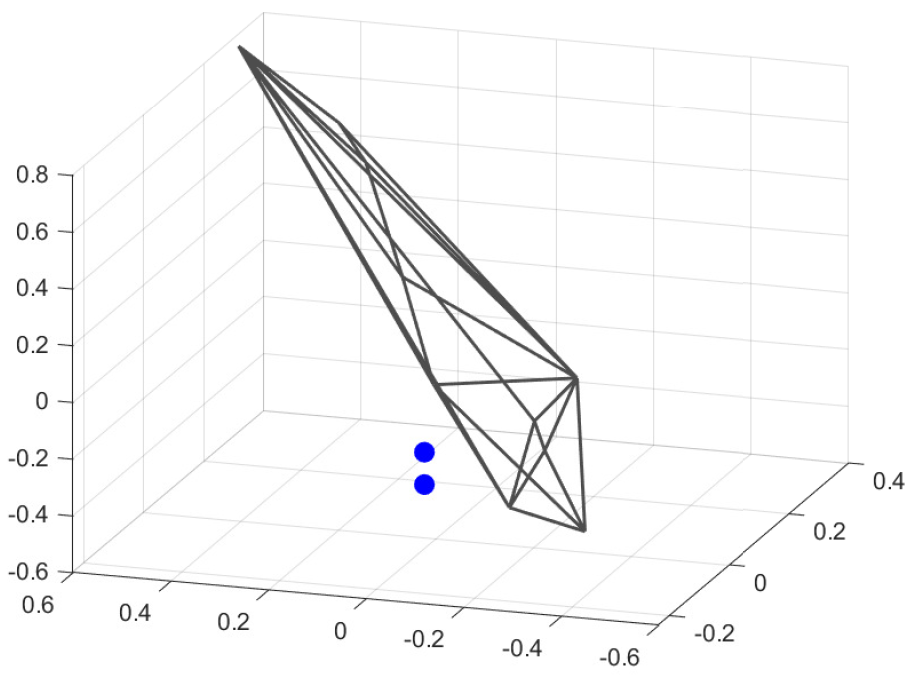}
\caption{$\calC$}
\end{subfigure}\;
\begin{subfigure}[b]{0.25\textwidth}
\centering
\includegraphics[scale=0.17]{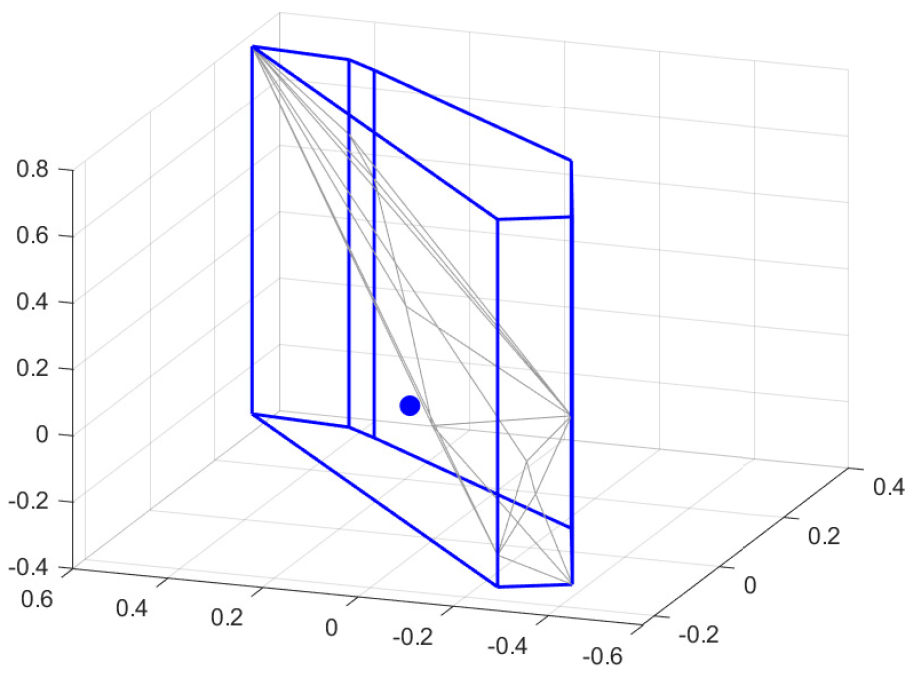}
\caption{$\calB_1$}
\end{subfigure}
\begin{subfigure}[b]{0.25\textwidth}
\centering
\includegraphics[scale=0.17]{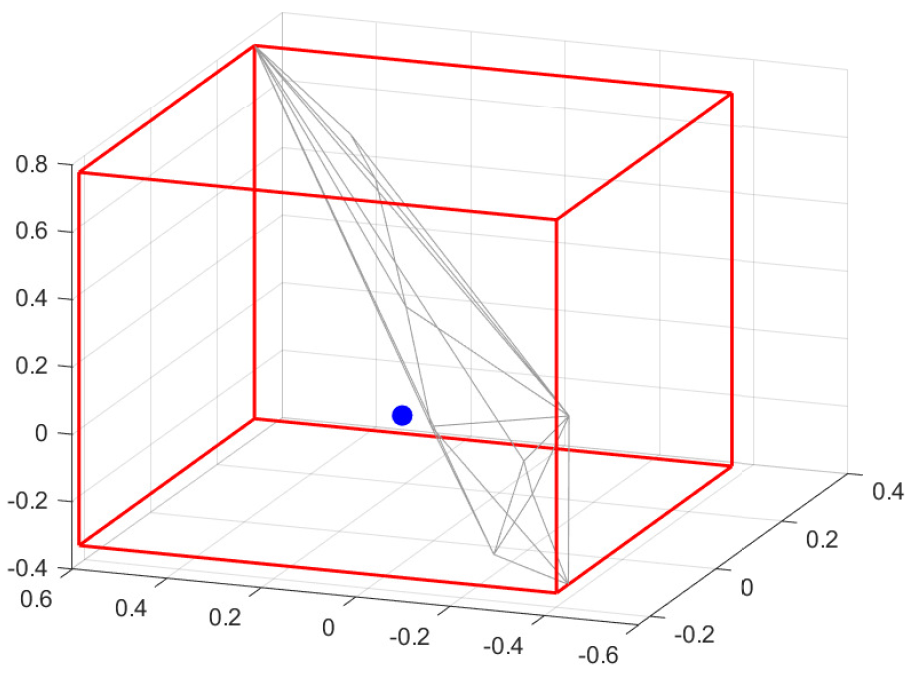}
\caption{$\calB_2$ }
\end{subfigure}
\caption{Convex hull of normal agents' initial positions, and regions in which normal agents converge to some point as a result of resilient bounded consensus.}
\label{fig:Num_Eval_Boxes}
\end{figure*}
% %------------- End Figure (Hulls) --------------------

% % %------------- Begin Figure (Hulls) ------------------
% \begin{figure}[htb]
% \centering
% \begin{subfigure}[b]{0.14\textwidth}
% \centering
% \includegraphics[scale=0.22]{HULL_CDC.eps}
% \caption{$\calC$}
% \end{subfigure}\;
% \begin{subfigure}[b]{0.14\textwidth}
% \centering
% \includegraphics[scale=0.22]{PRISM_CDC.eps}
% \caption{$\calB_1$}
% \end{subfigure}
% \begin{subfigure}[b]{0.14\textwidth}
% \centering
% \includegraphics[scale=0.22]{BOX_CDC.eps}
% \caption{$\calB_2$ }
% \end{subfigure}
% \caption{Convex hull of normal agents' initial positions, and regions in which normal agents converge to some point as a result of resilient bounded consensus.}
% \label{fig:Num_Eval_Boxes}
% \end{figure}
% % %------------- End Figure (Hulls) --------------------

% %------------- Begin Figure (Graphs) ------------------
% \begin{figure*}[ht]
% \centering
% \begin{subfigure}[b]{0.23\textwidth}
% \centering
% \includegraphics[scale=0.285]{GRAPH_1_CDC.eps}
% \caption{}
% \end{subfigure}\;
% \begin{subfigure}[b]{0.21\textwidth}
% \centering
% \includegraphics[scale=0.285]{HULL_CDC.eps}
% \caption{}
% \end{subfigure}
% \begin{subfigure}[b]{0.21\textwidth}
% \centering
% \includegraphics[scale=0.285]{PRISM_CDC.eps}
% \caption{}
% \end{subfigure}
% \begin{subfigure}[b]{0.21\textwidth}
% \centering
% \includegraphics[scale=0.285]{BOX_CDC.eps}
% \caption{}
% \end{subfigure}
% \caption{.}
% \label{fig:Num_Eval_Graphs}
% \end{figure*}
% %------------- End Figure (Graphs) --------------------

We illustrate resilient bounded consensus through an example. We consider a complete, undirected network of $N=20$ agents, of which $6$ are adversarial. The network graph is fixed and does not change over time. Since all agents are pairwise adjacent, each normal agent $i$ has $F_i=6$ adversarial agents in its neighborhood. Moreover, the state of each agent is its position in $\mathbb{R}^3$. Figure \ref{fig:Num_Eval_Graphs}(a) shows the network graph and initial positions of normal and adversarial agents, which are shown in blue and red colors, respectively. In our simulations, adversarial agents remain static and do not update their positions (states).

First, we implement 3-dimensional resilient consensus algorithm ADRC using centerpoint. Since for each normal agent $i$, we have $N_i=20$ and $F_i = 6 > \lceil\frac{N_i}{d+1}\rceil-1$, the resilience bound in \eqref{eqn:theory} is not satisfied. Figure \ref{fig:Num_Eval_Graphs}(b) shows the final positions of agents. We can see that normal agents fail to converge at a common point inside the convex hull of their positions, which is illustrated in Figure \ref{fig:Num_Eval_Boxes}(a). Second, we implement the resilient bounded consensus by partitioning the 3-dimensional state into 2-dimensional and scalar states. In other words, we use a partition $P=\{I_1,I_2\}$ of $I = \{1,2,3\}$ in Algorithm \ref{algo:RBC}, where $I_1=\{1,2\}$ and $I_2=\{3\}$. Since $F_i = 6 \le \lceil\frac{N_i}{\max_{\ell\in\{1,2\}}|I_\ell|+1}\rceil-1 = \lceil\frac{20}{3}\rceil-1$, convergence of normal nodes is guaranteed in a bounded region $\calB_1$. Figure \ref{fig:Num_Eval_Graphs}(c) shows the final positions of normal agents and Figure \ref{fig:Num_Eval_Boxes}(b) shows the bounded region $\calB_1$ along with the consensus point. Next, we implement the resilient bounded consensus by partitioning $3$-dimensional state into three $1$-dimensional states. Figure \ref{fig:Num_Eval_Graphs}(d) shows the final positions of agents. We observe that all normal agents achieve consensus and converge at a point inside a bounded region $\calB_2$, which is a hyperrectangle illustrated in Figure~\ref{fig:Num_Eval_Boxes}(c). 

We observe that the resilient bounded consensus in which the $3$-dimensional state is partitioned into lower-dimensional states has an improved resilience ($F_i = 6$) compared to the $3$-dimensional ADRC that is resilient against at most $F_i=4$ adversarial agents. However, this improvement comes at the cost of accuracy. The resilient bounded consensus only guarantees that normal agents converge in a bounded region $\calB_1 = \calC^1\times\calC^2$ and not necessarily inside the convex hull of initial positions $\calC$. Let $b_1\in\calB_1$ (similarly $b_2\in\calB_2$) be the convergence point of normal agents as a result of resilient bounded consensus algorithm. Then, the Hausdorff distance based accuracy bound in \eqref{eqn:Conjecture_3D} holds.
$$
\delta(b_1,\calC) = 0.18 \le \delta(\calB_1,\calC) = 0.731 \le \sqrt{\frac{3}{2}} \mu(\calC) = 1.9.
$$

Similarly, the bound in \eqref{eq:other_bound} also holds. 
$$
\delta(b_1,\calC) \le \delta(\calB_1,\calC) \le  \mu(\calC^1) = 1.13.
$$

Similarly, in case of $b_2\in\calB_2$, we have
$$
\delta(b_2,\calC) = 0.16 \le \delta(\calB_2,\calC) = 0.75 \le \sqrt{\frac{3}{2}} \mu(\calC) = 1.9.
$$

%At the same time, we get the volume ratios of $v(\calB_1,\calC) = 10.2$ and $v(\calB_2,\calC)=38.2$. 
Thus, as a result of resilient bounded consensus, the resilience improves and we are guaranteed to converge inside a bounded region at some point that is close to but not necessarily inside the convex hull of initial positions.

\section{Conclusion}
\label{sec:Con}
There is a trade-off between resilience and accuracy in the resilient multi-dimensional consensus problem. Resilience depends on the local connectivity of normal agents within the network, as well as the dimension of their state vector. If each normal agent $i$ has less than $\lceil\frac{N_i}{d+1}\rceil$ adversarial agents in its neighborhood, then all normal agents can be guaranteed to converge inside the convex hull $\calC$ of their initial states. We showed that the convergence of normal agents inside a bounded convex region $\calB \supseteq \calC$ can be guaranteed even if the number of adversaries in the neighborhood of a normal agent is more than $\lceil\frac{N_i}{d+1}\rceil$. For this, we partitioned $d$-dimensional state into multiple lower-dimensional states and implemented multiple instances of resilient consensus in lower dimensions. Since for a given $N_i$, resilience is better in lower dimensions, the overall resilience was improved. However, as a result of this, agents might converge outside of $\calC$ at some point in $\calB$. The maximum possible distance between the convergence point in $\calB$ and $\calC$ can be measured by the Hausdorff distance from $\calB$ to $\calC$. We provided upper bound on the Hausdorff distance from $\calB$ to $\calC$ in special cases. In the future, we would like to provide accuracy bounds for more general cases. Moreover, we would also like to explore other approaches to further exploit the trade-off between resilience and accuracy in resilient multi-dimensional consensus.
%============= References ========================
\bibliographystyle{IEEEtran}
\bibliography{refer}
\end{document}